\documentclass[11pt]{article}
\newcommand{\PAPER}[1]{#1}
\newcommand{\LIPICS}[1]{}

\PAPER{
\usepackage{fullpage,times,amsmath,amssymb,amsthm,hyperref,graphicx}

\newtheorem{theorem}{Theorem}[section]
\newtheorem{lemma}[theorem]{Lemma}
\newtheorem{corollary}[theorem]{Corollary}
}

\LIPICS{
\usepackage[utf8]{inputenc}
\usepackage{amsthm,amssymb,graphicx}
   \hypersetup{%
      breaklinks,%
      colorlinks=true,%
      urlcolor=[rgb]{0.25,0.0,0.0},%
      linkcolor=[rgb]{0.5,0.0,0.0},%
      citecolor=[rgb]{0,0.2,0.445},%
      filecolor=[rgb]{0,0,0.4},
      anchorcolor=[rgb]={0.0,0.1,0.2}%
   }
}

\newtheorem{problem}{Problem}

\newcommand{\R}{\mathbb{R}}

\newcommand{\down}[1]{\left\lfloor #1\right\rfloor}
\newcommand{\up}[1]{\left\lceil #1\right\rceil}

\newcommand{\QQ}{{\cal Q}}
\newcommand{\BB}{{\cal B}}
\newcommand{\SSS}{{\cal S}}
\newcommand{\nG}{n_0}

\title{Minimum $L_\infty$ Hausdorff Distance of Point Sets Under Translation: Generalizing Klee's Measure Problem}

\PAPER{
\author{Timothy M. Chan\thanks{Department of Computer Science,
University of Illinois at Urbana-Champaign (tmc@illinois.edu).  Work supported in part by NSF Grant CCF-2224271. }}
}

\LIPICS{
\author{Timothy M. Chan}{Department of Computer Science,
University of Illinois at Urbana-Champaign, USA}{tmc@illinois.edu}
{https://orcid.org/0000-0002-8093-0675}{Work supported in part by NSF Grant CCF-2224271.}

\titlerunning{Minimum $L_\infty$ Hausdorff Distance of Point Sets Under Translation}
\authorrunning{T.\,M. Chan}

\Copyright{Timothy M. Chan}



\EventEditors{Erin W. Chambers and Joachim Gudmundsson}
\EventNoEds{2}
\EventLongTitle{39th International Symposium on Computational Geometry (SoCG 2023)}
\EventShortTitle{SoCG 2023}
\EventAcronym{SoCG}
\EventYear{2023}
\EventDate{June 12--15, 2023}
\EventLocation{Dallas, Texas, USA}
\EventLogo{}
\SeriesVolume{258}
\ArticleNo{25}

\renewcommand{\paragraph}[1]{\subparagraph{#1}}

}

\begin{document}
\sloppy
\maketitle

\begin{abstract}
We present a (combinatorial) algorithm with running time close to $O(n^d)$ for computing the minimum directed $L_\infty$ Hausdorff distance between two sets of $n$ points under translations in any constant dimension $d$.  This substantially improves the best previous time bound near $O(n^{5d/4})$ by Chew, Dor, Efrat, and Kedem from more than twenty years ago.  Our solution is obtained by a new generalization of Chan's algorithm [FOCS'13] for Klee's measure problem.

To complement this algorithmic result, we also prove a nearly matching conditional lower bound close to $\Omega(n^d)$ for combinatorial algorithms, under the Combinatorial $k$-Clique Hypothesis.
\end{abstract}

\section{Introduction}\label{sec:intro}

\newcommand{\hhh}{\vec{h}}
This paper is about the following problem:

\begin{problem}\label{prob0}
\emph{($L_\infty$ Translational Hausdorff)}\ \
Given a set $P$ of $n$ points and a set $Q$ of $m$ points in $\R^d$,
compute
the minimum directed $L_\infty$ Hausdorff distance from $P$ to $Q$
under translation, i.e., compute $\min_{v\in\R^d} \hhh_\infty(P+v,Q)$
where $\hhh_\infty(P,Q):=\max_{p\in P}\min_{q\in Q}
\|p-q\|_\infty$.
\end{problem}

The problem has been extensively studied in computational geometry 
in the 1990s.  The analogous problem for \emph{undirected} Hausdorff distance
(defined
as $h_\infty(P,Q)=\max\{\hhh_\infty(P,Q),$ $\hhh_\infty(Q,P)\}$)
is reducible~\cite{Nusser21} to the directed version if $m=\Theta(n)$.  
The motivation lies in
measuring the resemblance between two geometric objects represented as
point clouds; furthermore, a connection with an even more fundamental
problem, Klee's measure problem
(see next page), provides added theoretical interest (and is what attracted
this author's attention in the first place).
Huttenlocher and Kedem~\cite{HuttenlocherK90}
introduced the problem and presented the first algorithms for $d=2$ (a subsequent paper~\cite{HuttenlocherKS93} also examined variants in $L_2$).
Chew and Kedem~\cite{ChewK98} gave an improved algorithm with
running time $O(mn\log^2(mn))$ for $d=2$ (in $L_\infty$), and generalized the algorithm to 
any constant dimension $d$
with running time $O((mn)^{d-1}\log^2(mn))$.
Chew, Dor, Efrat, and Kedem~\cite{ChewDEK99} 
described
further improved algorithms in higher dimensions: in the main case $m=n$,
their time bounds
were $O(n^3\log^2 n)$ for $d=3$, 
$O(n^{(4d-2)/3}\log^2n)$ for $4\le d\le 7$, and
$O(n^{5d/4}\log^2n)$ for any constant $d\ge 8$.  The exponent $5d/4$ looks peculiar, and
naturally raises the question of whether further improvements are still possible, but
none has been found in the intervening two decades  (except in
the logarithmic factors~\cite{Chan99,Chan13}).

Many other variants of Problem~\ref{prob0} have also been considered, for example,
using other metrics such as $L_2$ (as already mentioned above), 
allowing rotation and/or scaling besides translation,
handling other objects besides points,
allowing approximations, etc. (e.g., see \cite{ChewGHKKK97,AichholzerAR97,GoodrichMO99,IndykMV99,EfratIV04,ChoM08,AgarwalHSW10}). Several other alternatives
to the Hausdorff distance have also been popularly studied in computational geometry, such as the Earth mover distance and the Fr\'echet distance.  We will ignore all these
variants in the present paper, focusing only on exact directed $L_\infty$ Hausdorff
distance for point sets under translation.

Our new result is an algorithm for Problem~\ref{prob0} running in $O(n^d(\log\log n)^{O(1)})$ time
(using randomization) 
for any constant $d\ge 3$ in the main $m=n$ case (or with one extra logarithmic factor if randomization is not allowed).  The exponent $d$ is thus a \emph{substantial} improvement
over Chew et al.'s previous exponents for \emph{every} $d\ge 4$; see
Table~\ref{tbl1}.
In the general case, the running time of our algorithm is $O((mn)^{d/2}(\log\log(mn))^{O(1)})$.

\begin{table}
\begin{tabular}{l|llllllllllll}
dimension &$2$&$3$&$4$&$5$&$6$&$7$&$8$&$9$&$10$&$11$&$12$&$\cdots\!\!$\\\hline
prev.\ bound &$n^2$&$n^3$& $n^{4.66\cdots}\!\!$ & $n^6$
& $n^{7.33\cdots}\!\!$ & $n^{8.66\cdots}\!\!$ & $n^{10}$ & $n^{11.25}\!\!$
& $n^{12.5}\!$ & $n^{13.75}\!\!$ & $n^{15}$ &$\cdots\!\!$ \rule{0pt}{1.3em}\\
new bound & & &$n^4$&$n^5$&$n^6$&$n^7$&$n^8$&$n^9$&$n^{10}$
&$n^{11}$&$n^{12}$&$\cdots\!\!$
\end{tabular}
\caption{Previous upper bounds~\cite{ChewDEK99} and
new upper bounds for Problem~\ref{prob0} in the $m=n$ case, ignoring polylogarithmic factors.}
\label{tbl1}
\end{table}

\paragraph{Connection to a generalized Klee's measure problem.}
It suffices to focus on the decision problem: deciding whether the minimum is at most a given value $r$.
The original problem reduces to the decision problem, at the expense of
one extra logarithmic factor in the running time
by a well-known technique of Frederickson and Johnson~\cite{FredericksonJ84}
(in fact, when $d\ge 4$, a standard binary search suffices,
since the optimal value lies in a universe of $O((mn)^2)$ possible values which we can explicitly enumerate).  In some cases,
the extra logarithmic factor can even be eliminated by a randomized optimization technique~\cite{Chan99}.

Equivalently, we want to decide
whether there exists a vector $v\in\R^d$ with
$P+v\subseteq Q+[-r,r]^d$ (where ``$+$'' denotes the Minkowski sum when it
is clear from the context).  Assume (by rescaling) that $r=1/2$.
Let $\QQ$ be the set of unit hypercubes $\{q+[-1/2,1/2]^d: q\in Q\}$,
and let $S^*:= \bigcup_{B\in \QQ} B$.
The condition is equivalent to $P+v\subseteq S^*$,
i.e., $v\in \bigcap_{p\in P} (S^*-p)$.  Thus, the decision problem is equivalent to
the following:

\begin{problem}\label{prob1}
\emph{($L_\infty$ Translational Hausdorff Decision)}\ \ 
Given a set $P$ of $n$ points and
a set $\QQ$ of $m$ unit hypercubes\footnote{Throughout this paper, all hypercubes and boxes are
axis-aligned.} in $\R^d$,
decide whether $\bigcap_{p\in P} (S^*-p) = \emptyset$, where $S^*:=\bigcup_{B\in\QQ} B$.
\end{problem}

For each $B\in \QQ$ and $p\in P$, create a new unit hypercube $B-p$
and give this hypercube the color $p$.  Problem~\ref{prob1} then immediately reduces to the following
problem on $N = mn$ colored unit hypercubes: decide whether
$\bigcap_\chi S_\chi=\emptyset$, where $S_\chi := \bigcup_{\mbox{\scriptsize\rm $B\in\BB$ with color $\chi$}} B$.  (In other words, we want to decide whether
there exists a ``colorful'' point that lies in hypercubes of all colors.)

The unit hypercube case in turn reduces to the case of \emph{orthants} (i.e., $d$-sided boxes which are
unbounded in one direction along each axis):
we can build a uniform grid of unit-side length and solve the subproblem
inside each grid cell, but inside
a grid cell, a unit hypercube is identical to an orthant.  (We can ignore
grid cells that do not intersect hypercubes of all colors.)
Since a unit hypercube intersects only $O(1)$ grid cells, 
these subproblems have total input size $O(N)$.
All this motivates the definition of the following problem(s) on colored orthants:

\begin{problem}\label{prob2}
\emph{(Generalized Klee's Measure Problem)}\ \ 
Given a set $\BB$ of $N$ colored orthants in $\R^d$,

\begin{enumerate}
\item[\rm(a)]
decide whether $\bigcap_\chi S_\chi=\emptyset$,
\item[\rm(b)]
or more generally, compute a point of maximum or minimum depth among
the $S_\chi$'s (i.e., a point in the most or least number of regions $S_\chi$),
\item[\rm(c)] 
or alternatively, compute the volume of $\bigcap_\chi S_\chi$,
\end{enumerate}

\noindent
where $S_\chi := \bigcup_{\mbox{\scriptsize\rm $B\in\BB$ with color $\chi$}} B$.
\end{problem}

To recap, if Problem~\ref{prob2}(a) 
can be solved in $T(N)$ time, then Problem~\ref{prob1} can automatically be solved in 
$O(T(N))=O(T(mn))$ time (assuming superadditivity of $T(\cdot)$).

Problem~\ref{prob2}(a) is a generalization of the \emph{box coverage problem}: determine whether the union of a set of $N$ boxes in $\R^d$ covers
the entire space.  This is because $\bigcap_\chi S_\chi=\emptyset$ iff
$\bigcup_\chi \overline{S_\chi} = \R^d$, and
a box can be expressed as the complement of a union of at most $2d$ orthants
(we use $\overline{S}$ to denote the complement of a set $S$).
  Similarly, Problem~\ref{prob2}(b) is a generalization of the
\emph{box depth problem}: determine the minimum or maximum depth among
$N$ boxes in $\R^d$.  
Problem~\ref{prob2}(c) is a generalization
of \emph{Klee's measure problem}: compute the volume of the union of 
a set of $N$ boxes in $\R^d$.  
This generalization allows us to compute the volume of the union of more general shapes,
so long as each shape can be expressed as the complement of a union of orthants.
 (Note that we can clip to a bounding box to ensure that the volume is finite.)

The original Klee's measure problem has been extensively studied in computational
geometry \cite{Agarwal10,Bringmann12,Chan10,Chan13,OvermarsY91,YildizS15}.
The best known algorithm by Chan~\cite{Chan13} for Klee's measure
problem runs in $O(N^{d/2})$ time, based on a clever but simple
divide-and-conquer.  The box coverage and box depth
problem can be solved by similar algorithms, and in fact with slightly
lower time bounds by polylogarithmic factors using table lookup
and bit packing tricks~\cite{Chan13}.

For $d\le 3$, a union of orthants has linear combinatorial complexity~\cite{BoissonnatSTY95}
and can be constructed in near linear time.
Thus, a straightforward way to solve Problem~\ref{prob2} is to first construct
all the regions $S_\chi$ explicitly, decompose each $\overline{S_\chi}$
as a union of disjoint boxes, and then run a known algorithm
for Klee's measure problem on the resulting $O(N)$ boxes.
With this approach, Problem~\ref{prob2} can be solved in
$O(N\log N)$ time for $d=2$,
and $O(N^{3/2})$ time for $d=3$; consequently,
Problem~\ref{prob1} can be solved in $O((mn)\log(mn))$ time for $d=2$,
and $O((mn)^{3/2})$ time for $d=3$.  This was essentially how the
previous known 2D and 3D algorithms by Chew and Kedem~\cite{ChewK98}
and Chew et al.~\cite{ChewDEK99} were designed.

However, for $d\ge 4$, a union of $N$ orthants may have $\Theta(N^{\down{d/2}})$ combinatorial complexity in the worst case~\cite{BoissonnatSTY95}.  So, a two-stage
approach that explicitly constructs all the regions $S_\chi$ and then invokes
an algorithm for Klee's measure problem would be too slow!
Chew et al.~\cite{ChewDEK99} adapted
Overmars and Yap's algorithm for Klee's measure problem~\cite{OvermarsY91}
in a nontrivial way to obtain an $O(N^{5d/8}\log N)$-time algorithm
for Problem~\ref{prob2}, and consequently an $O(n^{5d/4}\log n)$-time algorithm
for Problem~\ref{prob1} when $m=n$.

We present a new algorithm that solves Problem~\ref{prob2}(c) in
$O(N^{d/2}\log^{d/2}N)$ time, matching
the known time bound 
for the original Klee's measure problem 
up to logarithmic factors.
In fact, for Problem~\ref{prob2}(a,b), the polylogarithmic factor
can be lowered to poly-$\log\log N$ factors using table lookup and
bit packing tricks.  Consequently, we obtain an $O(n^d (\log\log n)^{O(1)})$
time bound  for Problem~\ref{prob1} when $m=n$, or $O((mn)^{d/2}(\log\log(mn))^{O(1)})$ in general.
Our result is obtained by directly modifying Chan's divide-and-conquer algorithm for
Klee's measure problem~\cite{Chan13}.  The adaptation is not straightforward and uses interesting new ideas.  As mentioned,
we cannot afford to separate into two stages.  Instead, within
a single recursive process, we will 
handle two types of objects simultaneously, (i)~the input orthants, and
(ii)~features of the regions $S_\chi$ that have been found during the process.
The analysis of the recurrence is more delicate (though the overall algorithm remains simple).

\paragraph{Conditional lower bounds.}
In the other direction, recently in SoCG'21,
Bringmann and Nusser~\cite{BringmannN21}
proved an $\Omega((mn)^{1-\delta})$ conditional lower bound for Problem~\ref{prob0}--\ref{prob1} for $d=2$ for an arbitrarily small constant $\delta>0$, under the Orthogonal Vectors (OV) Hypothesis~\cite{virgisurvey}
(in particular, it holds under the Strong Exponential-Time Hypothesis (SETH)~\cite{virgisurvey}).
This showed that Chew and Kedem's upper bound for $d=2$ is likely near optimal~\cite{ChewK98}.  However, Bringmann and Nusser did not obtain any
lower bound in higher dimensions.

As observed by Chan~\cite{Chan10},
Klee's measure problem and the box coverage problem have 
an $\Omega(N^{d/2-\delta})$ lower bound for combinatorial algorithms
under the \emph{Combinatorial $k$-Clique Hypothesis}, which states that
there is no $O(\nG^{k-\delta})$-time combinatorial algorithm for
detecting a $k$-clique in a graph with $\nG$ vertices, for any constant $k\ge 3$.
The notion of ``combinatorial'' algorithms is not mathematically well-defined, but intuitively it refers to algorithms that avoid the use of fast matrix multiplication (such as Strassen's algorithm); all algorithms in this paper
and in Chew et al.'s previous paper fulfill this criterion.
(Recently, K\"{u}nnemann~\cite{Kun22} obtained new lower bounds 
for arbitrary, noncombinatorial algorithms for Klee's measure problem 
under the ``$k$-Hyperclique Hypothesis'', but
his bounds are not tight for $d\ge 4$.  See also \cite{CabelloGK08,BarequetH01,AronovC20}
for conditional lower bounds for other related geometric problems.
See \cite{JinX22} for a recent example of the usage of the Combinatorial $k$-Clique Hypothesis
in computational geometry, and
\cite{AbboudBW18,BringmannGL17} for other examples involving the Combinatorial $k$-Clique Hypothesis outside of geometry.)

Since our algorithms for Problem~\ref{prob2} have near $N^{d/2}$ running time,
they are near optimal for combinatorial algorithms under the
Combinatorial $k$-Clique Hypothesis.
However, this does not necessarily imply optimality of our algorithms for Problem~\ref{prob0} or~\ref{prob1}. 

In the second part of this paper, we prove that Problems \ref{prob0}--\ref{prob1}
have a conditional lower bound of $\Omega(n^{d-\delta})$ for 
$m=n$, or $\Omega((mn)^{d/2-\delta})$ for $m=n^\gamma$ for any fixed
$\gamma\le 1$, for combinatorial algorithms under the
Combinatorial $k$-Clique Hypothesis.  This shows that our combinatorial algorithm for Problem~\ref{prob0} is also conditionally near optimal.
While the previous conditional lower bound for Klee's measure problem by Chan~\cite{Chan10} was
obtained by reduction from $d$-clique,
we will reduce from clique of arbitrarily large constant size.
Our new reduction is more challenging and more interesting, but still simple.

\section{Algorithm}

\newcommand{\AAA}{{\cal E}}
\newcommand{\aaa}{E}

In this section, we present our new algorithm for the generalized Klee's measure problem (Problem~\ref{prob2})
for any constant dimension $d\ge 4$.
From this result, new algorithms for Problems~\ref{prob0}--\ref{prob1} will immediately follow.

To solve Problem~\ref{prob2}(c), we solve a generalization, where
we are given a box ``cell'' $\gamma$ and an extra set $\AAA$ of boxes, 
and we want to compute
the volume of $\bigcap_{\chi} S_\chi \cap\bigcap_{\aaa\in\AAA} \overline{\aaa}\cap \gamma$.  Initially, $\gamma=\R^d$ and $\AAA=\emptyset$.
We assume that the coordinates of the input have been pre-sorted (this
requires only an initial $O(N\log N)$ cost).

Call an orthant or a box \emph{short} if some of its $(d-2)$-faces intersect
$\gamma$'s interior, \emph{long} if it intersects $\gamma$'s interior
but is not short, and \emph{trivial} if it does not intersect $\gamma$'s
interior or it completely contains $\gamma$.

Our algorithm is inspired by Chan's divide-and-conquer algorithm~\cite{Chan13}
for the original Klee's measure problem, 
with many similarities (for example, in how we use weighted medians to divide a cell)
but also major new innovation (in how we reduce the number of long objects, and how
we ``convert'' some objects of $\BB$ into new objects in $\AAA$ during recursion).
The analysis of our algorithm requires a new charging argument and recurrence, causing some extra logarithmic factors.

\paragraph{Defining weights.}
Consider a $(d-2)$-face $f$ of a short orthant of $\BB$ such that
$f$ intersects $\gamma$'s interior.  If $f$ is orthogonal to 
the $i$-th and the $j$-th axes, assign $f$ a weight of 
$2^{(i+j)/d}$.  
Note that this weight is $\Theta(1)$, and so
each short orthant of $\BB$ contributes a total weight of $\Theta(1)$.

Similarly, consider a $(d-2)$-face $f$ of a box of $\AAA$ such that
$f$ intersects $\gamma$'s interior.  If $f$ is orthogonal to 
the $i$-th and the $j$-th axes, assign $f$ a weight of 
$2^{(i+j)/d}/t$, where $t\ge 1$ is a parameter to be set later.
Note that this weight is $\Theta(1/t)$, and so each short box of $\AAA$ contributes a total weight of
$\Theta(1/t)$.

\newcommand{\Nlong}{N_{\mbox{\scriptsize\rm long}}}
\newcommand{\Wshort}{W_{\mbox{\scriptsize\rm short}}}

Let $T(\Nlong,\Wshort)$ denote the time complexity
of the problem, where 
$\Nlong$ denotes the total number of long and trivial
orthants in $\BB$ and long and trivial boxes in $\AAA$,
and $\Wshort$ denotes
the total weight of all short orthants in $\BB$ and 
short boxes in $\AAA$. 
Note that the total number of orthants in $\BB$ and 
boxes in $\AAA$ is upper-bounded by $O(\Nlong+t\Wshort)$.

\begin{figure}
\centering
\includegraphics[scale=0.7]{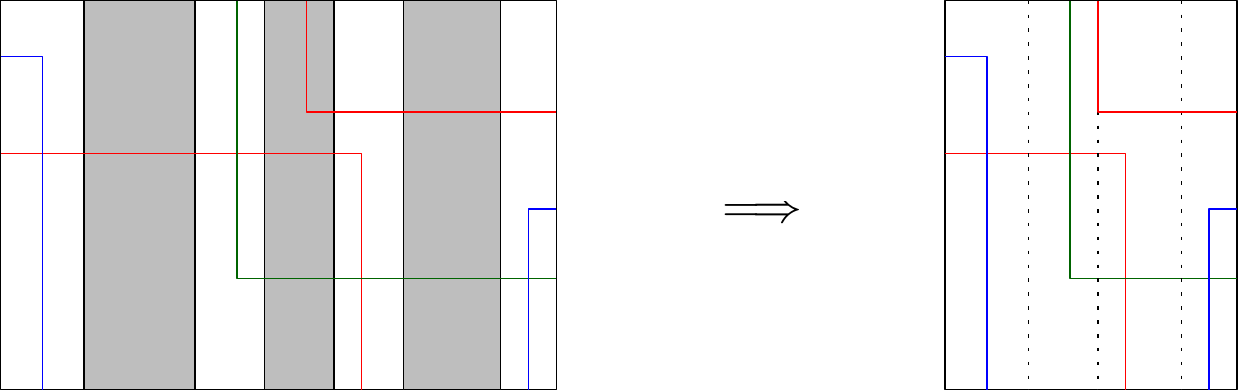}
\caption{Shrinking the width of the shaded slabs to 0, to eliminate long boxes in $\AAA$.}\label{fig:klee}
\end{figure}

\paragraph{Reducing the number of long objects.}
First, the trivial orthants and boxes can be easily eliminated:
We can remove all orthants of $\BB$ and boxes of $\AAA$ that do not intersect $\gamma$'s interior.
If an entire color class of $\BB$ does not intersect $\gamma$'s interior, or if some box of $\AAA$ completely contains $\gamma$,
we can return 0 as the answer.
If there is an orthant of $\BB$ completely containing $\gamma$, we can remove
its color class from $\BB$.

For each color~$\chi$, consider the long orthants of $\BB$ with color
$\chi$; the union of these long orthants are defined by
at most $2d$ orthants (since it is the complement of a box with at most $2d$ sides).  
Keep these $O(1)$ long orthants per color,
and remove the rest.
If there is a long orthant with color $\chi$ but
no short orthant with that color, then
$\overline{S_\chi}\cap\gamma$ is a box---add this box to $\AAA$
and remove the color class from $\BB$ (in other words, we have ``converted'' 
an entire color class in $\BB$ into a single box in $\AAA$).
This step increases $\Wshort$ by at most $\Nlong/t$.
After this step, each remaining long orthant can be ``charged'' to a short orthant of the same color, and so 
the number of remaining long orthants of $\BB$
is bounded by $O(1)$ times the number of short orthants, which is $O(\Wshort)$.

Next, for each $i\in\{1,\ldots,d\}$, consider the long boxes of $\AAA$
having $(d-1)$-faces intersecting $\gamma$ 
that are orthogonal to the $i$-th axis.  Compute the
union of these boxes by a linear scan after sorting, since this corresponds to computing the union of 1D intervals
when projected to the $i$-th axis.  The union forms a disjoint collection
of slabs.  Readjust all the $i$-th coordinates to
shrink the width of these slabs to 0, without altering
the volume of $\bigcap_{\chi} S_\chi \cap\bigcap_{\aaa\in\AAA} \overline{\aaa}\cap \gamma$, as illustrated in Figure~\ref{fig:klee}.
After doing this successively for every $i\in\{1,\ldots,d\}$,
all long boxes of $\AAA$ are eliminated.

After this process, there are $O(\Wshort)$ remaining long orthants of $\BB$ and zero long boxes of $\AAA$.  Thus, $\Nlong$ is reduced to $O(\Wshort)$.  The weight of the short orthants of $\BB$ may increase to at most $\Wshort+\Nlong/t$.  We then have
the following, for some constant~$c$:
\begin{equation}\label{eqn:reduce}
T(\Nlong,\Wshort) \:\le\: T(c\Wshort, \Wshort+\Nlong/t) + O(\Nlong + t\Wshort).
\end{equation}

\paragraph{Divide-and-conquer.}
Next, compute the weighted median $m$ of the first coordinates
of the $(d-2)$-faces of $\BB$ and $\AAA$ intersecting $\gamma$'s interior
that are orthogonal to the first axis.
Divide $\gamma$ into two subcells $\gamma_L$ and $\gamma_R$
by the hyperplane $\{(x_1,\ldots,x_d): x_1=m\}$.
Renumber the coordinate axes $1,2,\ldots,d$ to $2,\ldots,d,1$,
and recursively solve the problem for $\gamma_L$ and for $\gamma_R$.

To analyze the algorithm, consider 
a $(d-2)$-face $f$ of $\BB$ (resp.\ $\AAA$) orthogonal to the $i$-th and $j$-th axes
with $i,j\neq 1$.  After the axis renumbering, its weight
changes from $2^{(i+j)/d}$ to $2^{(i-1+j-1)/d}$ (resp.\ from $2^{(i+j)/d}/t$ to $2^{(i-1+j-1)/d}/t$), i.e., the weight
decreases by a factor of $2^{2/d}$.  

Next consider 
a $(d-2)$-face $f$ of $\BB$ (resp.\ $\AAA$) orthogonal to the first and the $j$-th axes
with $j\neq 1$.  After the axis renumbering, its weight
changes from $2^{(1+j)/d}$ to $2^{(d+j-1)/d}$ (resp.\ from $2^{(1+j)/d}/t$ to $2^{(d+j-1)/d}/t$), i.e., the weight
increases by a factor of $2^{(d-2)/d}$.
But when $\gamma$ is divided into subcells $\gamma_L$ and $\gamma_R$,
the weight within each subcell decreases by a factor of 2; the net
decrease in weight is thus a factor of $2^{2/d}$.

Hence, $\Wshort$ decreases by a factor of $2^{2/d}$ in either subcell.  (On the other hand,
$\Nlong$ may not necessarily decrease.)
We then have
\begin{equation}\label{eqn:divide}
 T(\Nlong,\Wshort)\:\le\: 2\,T(\Nlong,\Wshort/2^{2/d}) + O(\Nlong + t\Wshort).
\end{equation}

\paragraph{Putting it together.}
By combining (\ref{eqn:divide}) and (\ref{eqn:reduce}) 
and letting $T(N) := T(cN,N)$, we obtain
\begin{eqnarray*}
T(N)\ \le\ 2\,T(cN, N/2^{2/d}) + O(tN) 
&\le& 2\, T(cN/2^{2/d}, N/2^{2/d}+cN/t) + O(tN)\\
&\le & 2\, T\left(\tfrac{1+ 2c/t}{2^{2/d}}N\right) + O(tN).
\end{eqnarray*}
For the base case, we have $T(O(1))=O(t^{d/2})$: when
$\Wshort=O(1)$, there are
$O(1)$ orthants of $\BB$ and $O(t)$ boxes of $\AAA$, and so the problem
can be solved by running a known algorithm for Klee's measure problem
on $O(t)$ boxes~\cite{Chan13}.

By the master theorem,
the solution to the recurrence is
\[ T(N)\ =\ O(t^{d/2} N^{1/\log_2 (2^{2/d}/(1+2c/t))})
\ =\ O(t^{d/2} N^{d/2 + O(1/t)}).
\]
Choosing $t=\log N$ yields $T(N)=O(N^{d/2}\log^{d/2}N)$.
This completes the description and analysis of the main algorithm.

\paragraph{Shaving logs by bit packing.}
For Problem~\ref{prob2}(a), we can obtain a minor (but not-very-practical) improvement
in the polylogarithmic factors by using more technical but standard
bit-packing tricks, as we now briefly explain (see \cite{Chan13} for more details on
these kinds of tricks):  The main observation is that
actual coordinate values do not matter here,
only their relative order, so we can replace them with their ranks
in the sorted list.
Thus, the $O(\Nlong + t\Wshort)$ input objects can be represented
by $O((\Nlong + t\Wshort)\log(\Nlong + t\Wshort))$ bits and packed in
$O(((\Nlong + t\Wshort)\log(\Nlong + t\Wshort))/w)$ words, assuming
a $w$-bit word RAM model of computation.
The $O(\Nlong + t\Wshort)$ cost for various linear scans during 
recursion can be reduced to $O(((\Nlong + t\Wshort)\log^2(\Nlong + t\Wshort))/w)$, since sorting $k$ $b$-bit numbers
can be done in $O((kb\log k)/w)$ time by a packed version of mergesort.
Thus, the recurrence changes to 
$T(N)\:\le\: 2\, T(\tfrac{1+ 2c/t}{2^{2/d}}N) + O(1 + (tN\log^2 (tN))/w).$

For the base case, we can use a bit-packed version of Chan's algorithm
for the box coverage problem on $O(t)$ boxes, which runs in $T(O(1))=O(1+(t/w)^{d/2}\log^{O(1)}w)$ time~\cite[Section~3.1]{Chan13}.
The recurrence solves to
$T(N)=O(N^{d/2+O(1/t)} \cdot (1+(t/w)^{d/2}\log^{O(1)}w))$.
Choosing $t=\log N$ yields $T(N) = O(N^{d/2}\cdot(1+ ((\log N)/ w)^{d/2}\log^{O(1)}w))$.

The above may require nonstandard operations on $w$-bit words.
By choosing $w=\delta\log N$ for a sufficiently small constant $\delta>0$,
such operations can be simulated in constant time via table lookup
after preprocessing in $2^{O(w)} = N^{O(\delta)}$ time.
Hence, we obtain the final time bound of $O(N^{d/2}(\log\log N)^{O(1)})$.

Problem~\ref{prob2}(b) can be solved similarly, with some modification
to the steps to reduce the number of long boxes in $\AAA$, but this is identical
to the modification of Chan's algorithm for the original
box depth problem~\cite[Section~3.1]{Chan13}.  To summarize, we have obtained the following theorem:

\begin{theorem}
For any constant $d\ge 4$,
Problem~\ref{prob2}(c) can be solved in $O(N^{d/2}\log^{d/2}N)$ time,
and Problems~\ref{prob2}(a,b) can be solved in $O(N^{d/2}(\log\log N)^{O(1)})$ time.
\end{theorem}

\begin{corollary}
For any constant $d\ge 4$,
Problem~\ref{prob1} can be solved in $O((mn)^{d/2}(\log\log(mn))^{O(1)})$ time.
Problem~\ref{prob0}  can be solved in $O((mn)^{d/2}\log (mn) (\log\log (mn))^{O(1)})$ time deterministically, or in 
$O((mn)^{d/2}(\log\log(mn))^{O(1)})$ expected time with randomization. 
\end{corollary}
\begin{proof}
As mentioned in Section~\ref{sec:intro}, Problem~\ref{prob0} reduces to
Problem~\ref{prob1} by Frederickson and Johnson's technique~\cite{FredericksonJ84} or by ordinary binary search,
with an extra logarithmic factor.

Chan~\cite[Section~4.2]{Chan99} has described how to apply his
randomized optimization technique to reduce
the following problem to its decision problem without losing
a logarithmic factor: 
\begin{quote}
Given $N$ colored points in $\R^d$,
find the smallest hypercube that contains points of all colors.
\end{quote}
As noted in~\cite{Chan99}, Problem~\ref{prob0}
reduces to this problem.  On the other hand, as noted 
in Section~\ref{sec:intro}, the decision version
of this problem (equivalent to finding a point that is inside unit hypercubes of all colors) reduces to Problem~\ref{prob2}(a), which we have just solved.
\end{proof}

\section{Conditional Lower Bound}

\newcommand{\XX}{{\cal Z}}

In this section, we prove a nearly matching conditional lower bound for Problems~\ref{prob0}--\ref{prob1} for combinatorial algorithms
under the Combinatorial $k$-Clique Hypothesis.
We first introduce an intermediate problem which is more convenient to work with.
Roughly speaking, Problem~\ref{prob1} considers the intersection of  translates
of a \emph{single} shape (the shape being a union of unit hypercubes), whereas
the problem below considers the intersection of translates of \emph{multiple} shapes (each shape being
a union of orthants).

\begin{problem}\label{prob3}
Let $\XX$ be a set of shapes, where each shape is
a union of orthants in $\R^d$.  Let $m$
be the total number of orthants over all shapes of $\XX$.
Given a set $\SSS$ of
$n$ objects where each object is a translate of some shape in $\XX$,
decide whether $\bigcap_{S\in\SSS}S=\emptyset$.
\end{problem}

\begin{lemma}\label{lem:prob3}
Problem~\ref{prob3} reduces to Problem~\ref{prob1} on $O(n)$ points
and $O(m)$ unit hypercubes.
\end{lemma}
\begin{proof}
Assume  (by rescaling) that the coordinates of all orthants of $\XX$ and all translation
vectors used in $\SSS$ are in $[0,1/2]$.  In particular, if $\bigcap_{S\in\SSS}S$
is nonempty, it must contain a point in $[0,1]^d$.
Inside $[0,1]^d$, each orthant of $\XX$ may be replaced by an equivalent unit hypercube.

Let $Z_1,\ldots,Z_\ell$ be the shapes of $\XX$.
Let $\BB_i$ be the unit hypercubes corresponding to the orthants defining
$Z_i$ (so that $(Z_i+t)\cap [0,1]^d = \bigcup_{B\in \BB_i} (B+t) \cap [0,1]^d$ for any $t\in [0,1/2]^d$).

We construct an instance of Problem~\ref{prob1} as follows:
For each $B\in \BB_i$, add the shifted unit hypercube
$B+u_i$ to $\QQ$ where $u_i := (4i,0,\ldots,0)\in\R^d$.
(This operation distributes objects in different classes $\BB_i$ to different parts
of space, since the vectors $u_i$ are at least 4 units apart from each other along the first axis.)
For each object $S\in\SSS$, if $S$ is the translate $Z_i+t$,
add the point $u_i-t$ to $P$.  Lastly, add two auxiliary unit
hypercubes $[0,1]^d$ and
$u_{\ell+1}+[0,1]^d$ to $\QQ$, and two auxiliary points $u_0$ and $u_{\ell+1}$ to $P$.

We solve Problem~\ref{prob1} on these points of $P$ and these unit hypercubes of $\QQ$, to determine whether $\bigcap_{p\in P}(S^*-p)=\emptyset$, where $S^*:=\bigcup_{B\in\QQ}B$.  For correctness, we just observe that $\bigcap_{p\in P}(S^*-p)$ is identical to $\bigcap_{S\in\SSS}S$ inside $[0,1]^d$.
This is because for each $B\in\BB_i$, $(B+u_i)-(u_{i'}-t)$ may intersect $[0,1]^d$ only if $i=i'$
(since $u_i$ and $u_{i'}$ are far apart if $i\neq i'$), assuming $t\in [0,1/2]^d$.
\end{proof}

We now prove hardness of Problem~\ref{prob3} by reduction from the
clique problem for graphs.  We first warm up with two simpler reductions
yielding weaker lower bounds, before presenting the final reduction in 
Lemma~\ref{lem:final}.  
(Readers who do not need intuition building may
go straight to Lemma~\ref{lem:final}'s proof.)

\paragraph{First attempt.}
First observe that the box coverage problem 
(deciding whether $n$ boxes in $\R^d$ cover the entire space, i.e., deciding whether the intersection of the complements of $n$ boxes is empty)
easily reduces to
Problem~\ref{prob3} with $n$ orthants and $n$ objects in $\R^d$,
since the complement of a box is the union of $O(1)$ orthants.
By Lemma~\ref{lem:prob3}, we 
immediately obtain an $\Omega(n^{d/2-\delta})$ conditional lower bound for Problem~\ref{prob3}, since 
the box coverage problem has an $\Omega(n^{d/2-\delta})$ lower bound under the Combinatorial $k$-Clique Hypothesis~\cite{Chan10}.

In the following lemma, we directly modify the (very simple) known reduction from clique to the box coverage
problem~\cite{Chan10}, to show the same lower bound even when the number of orthants $m$ is $O(1)$:

\begin{lemma}
Detecting a $d$-clique in a graph with $\nG$ vertices
reduces to Problem~\ref{prob3} with $m=O(1)$ orthants and
$n=O(\nG^2)$ objects in $\R^d$.
\end{lemma}
\begin{proof}
Let $G=(V,E)$ be the given graph, with $V=[\nG]=\{0,\ldots,\nG-1\}$.
We will construct a set $\SSS$ of objects whose intersection is
\begin{equation}\label{eqn0}
\{(x_1,\ldots,x_{d})\in [0,\nG)^d:\ \{\down{x_1},\ldots,\down{x_d}\}\ \mbox{is
a $d$-clique of $G$}\}.
\end{equation}
It would then follow that the intersection is nonempty iff
a $d$-clique exists.

The construction is very simple: for each $\alpha,\beta\in \{1,\ldots,d\}$
with $\alpha\neq\beta$
and for each $u,v\in [\nG]$ with $uv\not\in E$, 
add the complement of the box
\[ B_{\alpha,\beta,u,v}\::=\: \{(x_1,\ldots,x_d):\ \down{x_\alpha}=u,\ \down{x_\beta}=v\} 
\]
to $\SSS$.  (Note that if $u=v$, we consider $uv\not\in E$.)
These boxes are unit squares when projected to the $\alpha$-th
and $\beta$-th axes, and are thus translates of $O(1)$ fixed boxes, and the
complement of each such fixed box can obviously be expressed as a union of $O(1)$ orthants and can be added to $\XX$.
Lastly, add the $O(1)$ halfspaces bounding $[0,\nG)^d$ to $\SSS$.
Then $\SSS$ has a total of $O(\nG^2)$ translates and clearly satisfies the desired property (\ref{eqn0}).
\end{proof}

In combination with Lemma~\ref{lem:prob3},
the above lemma indeed implies an $\Omega(n^{d/2-\delta})$ conditional lower bound for Problem~\ref{prob1}:
if Problem~\ref{prob1} for $m=O(1)$ has an $O(n^{d/2-\delta})$-time combinatorial
algorithm, then the $d$-clique detection problem for a graph with $\nG$ vertices
has a combinatorial algorithm with running time $O((\nG^2)^{d/2-\delta})$
= $O(\nG^{d-2\delta})$, contradicting the Combinatorial $k$-Clique Hypothesis.

\paragraph{Second attempt.}
We now improve the lower bound by reducing from clique of a large size $2d$.
We use  the following key idea: encode a pair of vertices in a single coordinate value.

\begin{figure}
\centering
\includegraphics[scale=1]{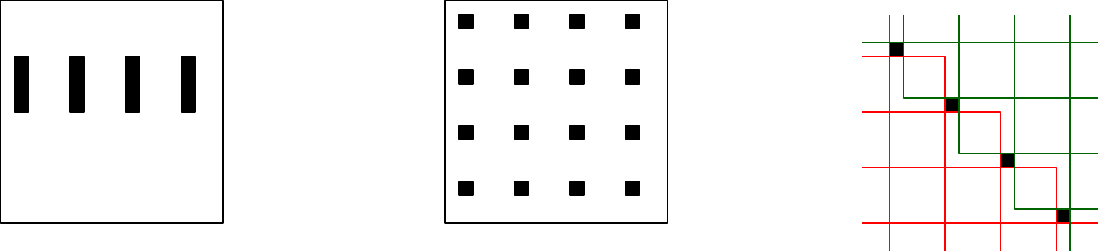}
\caption{(left) A region $Y_{\alpha,\beta,a,b,u,v}$ for $(a,b)=(0,1)$.
(middle) A region $Y_{\alpha,\beta,a,b,u,v}$ for $(a,b)=(0,0)$.
(right) A diagonal $D_{\alpha,\beta}$ (whose complement can be expressed as a union of the red and the green orthants).}\label{fig:grid}
\end{figure}

\begin{lemma}
Detecting a $(2d)$-clique in a graph $G$ with $\nG$ vertices
reduces to Problem~\ref{prob3} with $m=O(\nG)$ orthants and
$n=O(\nG^3)$ objects in $\R^d$.
\end{lemma}
\begin{proof}
Let $G=(V,E)$ be the given graph, with $V=[\nG]$.
For any $x\in [0,\nG^2)$, 
let $\phi_0(x) = \down{x}\bmod{\nG}$ and $\phi_1(x) = \down{x/\nG}$.
We will construct a set $\SSS$ of objects whose intersection is
\begin{equation}\label{eqn1}
 \{(x_1,\ldots,x_{d})\in [0,\nG^2)^d:\ \{\phi_0(x_1),\phi_1(x_1),\ldots,
\phi_0(x_d),\phi_1(x_d)\}\ \mbox{is
a $(2d)$-clique in $G$}\}.
\end{equation}
It would then follow that the intersection is nonempty iff
a $(2d)$-clique exists.

For each $\alpha,\beta\in \{1,\ldots,d\}$ and $a,b\in\{0,1\}$
with $(\alpha,a)\neq (\beta,b)$, and for each $u,v\in [\nG]$
with $uv\not\in E$, 
define the region
\[
Y_{\alpha,\beta,a,b,u,v} \::=\:  
 \{(x_1,\ldots,x_{d})\in [0,\nG^2)^d:\  \phi_a(x_\alpha)=u,\ \phi_b(x_\beta)=v\}.
\]

If $\alpha=\beta$, then $Y_{\alpha,\beta,a,b,u,v}$ is just a unit interval
when projected to the $\alpha$-th axis, and is thus a translate of one
of $O(1)$ fixed boxes,  
and the complement of each such fixed box can be expressed as a union of $O(1)$ orthants and can
be added to $\XX$.  From now on, assume $\alpha\neq\beta$.

If $(a,b)=(1,1)$, then $Y_{\alpha,\beta,a,b,u,v}$ is just an $\nG\times \nG$ square
when projected to the $\alpha$-th and $\beta$-th axes,
and is thus a translate of one
of $O(1)$ fixed boxes,
and the complement of each such fixed box can be expressed as a union of $O(1)$
orthants and can be added to $\XX$.

If $(a,b)=(0,1)$ (or $(a,b)=(1,0)$), then $Y_{\alpha,\beta,a,b,u,v}$ is a union of
$\nG$ rectangles of dimension $1\times \nG$ (or $\nG\times 1$) when projected to the $\alpha$-th and $\beta$-th axes (see Figure~\ref{fig:grid}(left)), and is thus a union of $\nG$ translates of
one of $O(1)$ fixed boxes,
and the complement of each fixed box can be expressed as a union of $O(1)$
orthants and can be added to $\XX$.  

If $(a,b)=(0,0)$, then $Y_{\alpha,\beta,a,b,u,v}$ forms a $\nG\times\nG$ grid pattern
when projected to the $\alpha$-th and $\beta$-th axes  (see Figure~\ref{fig:grid}(middle)).
Although the complement of this region
can't be expressed as a union of orthants, we can decompose the grid into subregions that can.  The most obvious approach is to decompose into rows or columns, but this
still doesn't work.  Instead, we will decompose into ``diagonals''.
More precisely, define
\begin{eqnarray*} D_{\alpha,\beta}\ :=\ \{(x_1,\ldots,x_d):\!\!&& \down{x_\alpha}\bmod{\nG}=0,\ 
\down{x_\beta}\bmod{\nG}=0,\\ && \down{x_\alpha/\nG}+\down{x_\beta/\nG}=\nG,\ 
x_\alpha,x_\beta\ge 0\}.
\end{eqnarray*}
Since $D_{\alpha,\beta,a,b}$ can be viewed as the region sandwiched between
two staircases when projected to 2D,
its complement $\overline{D_{\alpha,\beta,a,b}}$ can be expressed as a union of
$O(\nG)$ orthants  (see Figure~\ref{fig:grid}(right)).  Add the shape $\overline{D_{\alpha,\beta,a,b}}$ to $\XX$.
The region $Y_{\alpha,\beta,a,b,u,v}$ can be expressed as a union of
$O(\nG)$ translates of $D_{\alpha,\beta}$ when clipped to $[0,\nG^2)^d$.

In any case,
$\overline{Y_{\alpha,\beta,a,b,u,v}}\cap[0,\nG^2)^d$ can be expressed as the intersection
of $[0,\nG^2)^d$ with $O(\nG)$ translates of shapes from $\XX$.
Add these translates to $\SSS$,
for each $\alpha,\beta\in \{1,\ldots,d\}$ and $a,b\in [g]$ with $(\alpha,a)\neq (\beta,b)$, and for each $u,v\in [\nG]$ with
$uv\not\in E$. 
Lastly, add the $O(1)$ halfspaces bounding $[0,\nG^2)^d$ to $\SSS$.
Then $\SSS$ has a total of $O(\nG^3)$ translates and satisfies the desired property~(\ref{eqn1}).
\end{proof}

The above lemma implies a larger $\Omega(n^{2d/3-\delta})$ conditional lower bound:
if Problem~\ref{prob1} for $m=n^{1/3}$ has an $O(n^{2d/3-\delta})$-time combinatorial
algorithm, then the $(2d)$-clique detection problem for a graph with $\nG$ vertices
has a combinatorial algorithm with running time $O((\nG^3)^{2d/3-\delta})$
= $O(\nG^{2d-3\delta})$, contradicting the Combinatorial $k$-Clique Hypothesis.

\paragraph{Final reduction.}
We obtain our final lower bound by generalizing the idea further.  We reduce
from clique of still larger size and now encode $g$-tuples of vertices instead of pairs
(incidentally, the idea of encoding tuples has also
appeared recently in K\" unnemann's conditional lower bound proofs for Klee's
measure problem~\cite{Kun22}):

\begin{lemma}\label{lem:final}
Let $g$ be any integer constant.
Detecting a $(dg)$-clique in a graph $G$ with $\nG$ vertices
reduces to Problem~\ref{prob3} with $m=O(\nG^{g-1})$ orthants and
$n=O(\nG^{g+1})$ objects in $\R^d$.

More generally, for any given $m\le \nG^{g-1}$,
detecting a $(dg)$-clique in a graph $G$ with $\nG$ vertices
reduces to Problem~\ref{prob3} with $m$ orthants and
$n=O(\nG^{2g}/m)$ objects in $\R^d$.
\end{lemma}
\begin{proof}
Let $G=(V,E)$ be the given graph, with $V=[\nG]$.
For any $x\in [0,\nG^g)$ and $a\in [g]$, 
let $\phi_a(x)$ be the $(a+1)$-th least significant digit of $\down{x}$ in
base $n$.
We will construct a set $\SSS$ of objects whose intersection is
\begin{eqnarray}
\{(x_1,\ldots,x_{d})\in [0,\nG^g)^d:\!\! && \{\phi_0(x_1),\ldots,\phi_{g-1}(x_1),\ldots,\phi_0(x_{d}),\ldots,\phi_{g-1}(x_{d})\}\nonumber\\ && \mbox{is a $(dg)$-clique in $G$}\}.\label{eqn3}
\end{eqnarray}
It would then follow that the intersection is nonempty iff
a $(dg)$-clique exists.

For each $\alpha,\beta\in \{1,\ldots,d\}$ and $a,b\in [g]$ with $(\alpha,a)\neq(\beta,b)$, and for each $u,v\in [\nG]$ with $uv\not\in E$, 
define the region
\begin{eqnarray*}
Y_{\alpha,\beta,a,b,u,v} &:=&  
 \{(x_1,\ldots,x_{d})\in [0,\nG^d)^d:\  \phi_a(x_\alpha)=u,\ \phi_b(x_\beta)=v\}\\[1ex]
&=& \{(x_1,\ldots,x_{d}):\ x_\alpha\in [i\nG^{a+1}+u\nG^a,i\nG^{a+1}+(u+1)\nG^a),\\[-2pt]
&&\qquad\qquad\qquad\ \ \ x_\beta\in [j\nG^{b+1}+v\nG^b,jn^{b+1}+(v+1)\nG^b)\\[-2pt]
&&\qquad\qquad\qquad\ \ \ \mbox{for some $i\in [\nG^{g-a-1}]$, $j\in [\nG^{g-b-1}]$}\}.
\end{eqnarray*}

If $\alpha=\beta$, then $Y_{\alpha,\beta,a,b,u,v}$ is a union of
at most $O(\nG^{g-1})$ unit intervals when projected to the $\alpha$-th axis,
and is thus a union of $O(\nG^{g-1})$ translates of $O(1)$ fixed
boxes, and the complement of each fixed box can be expressed
as a union of $O(1)$ orthants and can be added to $\XX$.  

If $\alpha\neq\beta$, then $Y_{\alpha,\beta,a,b,u,v}$ forms a $O(\nG^{g-1})\times O(\nG^{g-1})$ grid pattern
when projected to the $\alpha$-th and $\beta$-th axes.
Define the ``diagonal''
\[
\begin{array}{ll}
D_{\alpha,\beta,a,b}\ =\ \{(x_1,\ldots,x_{d}): & x_\alpha\in [i\nG^{a+1},i\nG^{a+1}+\nG^a),\ x_\beta\in [j\nG^{b+1},j\nG^{b+1}+\nG^b)\\
& \mbox{for some $i,j\in [m]$ with $i+j=m$}\}.
\end{array}
\]
Since $D_{\alpha,\beta,a,b}$ can be viewed as the region sandwiched between
two staircases when projected to 2D, 
its complement $\overline{D_{\alpha,\beta,a,b}}$ can be expressed as a union of
$O(m)$ orthants.  Add the shape $\overline{D_{\alpha,\beta,a,b}}$ to $\XX$.
The region $Y_{\alpha,\beta,a,b,u,v}$ can be expressed as a union of
$O(\nG^{2(g-1)}/m)$ translates of $D_{\alpha,\beta,a,b}$ when clipped to $[0,\nG^g)^d$.

In any case, $\overline{Y_{\alpha,\beta,a,b,u,v}}\cap [0,\nG^g)^d$ can be expressed
as an intersection of $[0,\nG^g)^d$ with $O(\nG^{2(g-1)}/m)$ translates of shapes from $\XX$.
Add all these translates to $\SSS$, for each $\alpha,\beta\in \{1,\ldots,d\}$ and $a,b\in [g]$ with $(\alpha,a)\neq (\beta,b)$, and for each $u,v\in [\nG]$ with
$uv\not\in E$.  
Lastly, add the $O(1)$ halfspaces bounding $[0,\nG^g)^d$ to $\SSS$.
Then $\SSS$ has a total 
of $O(\nG^2\cdot \nG^{2(g-1)}/m)=O(\nG^{2g}/m)$ translates 
and satisfies the desired property (\ref{eqn3}).
\end{proof}

The above lemma implies an $\Omega(n^{gd/(g+1)-\delta})$ conditional lower bound
for any integer constant $g$:
if Problem~\ref{prob1} for $m=n^{(g-1)/(g+1)}$ has an $O(n^{gd/(g+1)-\delta})$-time combinatorial
algorithm, then the $(dg)$-clique detection problem for a graph with $\nG$ vertices
has a combinatorial algorithm with running time $O((\nG^{g+1})^{gd/(g+1)-\delta})$
= $O(\nG^{dg-(g+1)\delta})$, contradicting the Combinatorial $k$-Clique Hypothesis.
The exponent $gd/(g+1)-\delta$ exceeds $d-2\delta$, by picking a sufficiently large $g\ge d/\delta$. 

More generally, for any constant $\gamma\le (g-1)/(g+1)$, if Problem~\ref{prob1} for $m=n^\gamma$ has an $O((mn)^{d/2-\delta})$-time combinatorial
algorithm, then the $(dg)$-clique detection problem for a graph with $\nG$ vertices
has a combinatorial algorithm with running time $O(((\nG^{2g}/m)\cdot m)^{d/2-\delta})$
= $O(\nG^{dg-2g\delta})$, contradicting the Combinatorial $k$-Clique Hypothesis.

\begin{theorem}
Under the Combinatorial $k$-Clique Hypothesis,
Problem~\ref{prob0} or \ref{prob1} for $n$ points and $n$ unit hypercubes in $\R^d$ 
does not have an $O(n^{d-\delta})$-time combinatorial algorithm for any constant $\delta>0$.

More generally, under the same hypothesis, for any fixed constant $\gamma\le 1$, Problem~\ref{prob0} or~\ref{prob1} for $m=n^\gamma$ points and $n$ unit hypercubes in $\R^d$ 
does not have an $O((mn)^{d/2-\delta})$-time combinatorial algorithm for any constant $\delta>0$.
\end{theorem}

\section{Final Remarks}

To summarize, we have studied  the $L_\infty$ translational Hausdorff distance 
problem for point sets, a fundamental problem with a long history in computational geometry.
We have obtained a substantially improved upper bound for this problem, and the first
conditional lower bound in dimension 3 and higher, which nearly match the upper bound.
Our technique for the upper bound is interesting, in that it implies a natural colored generalization
of Klee's measure problem 
can be solved in roughly the same time bound as the original Klee's problem.
Our lower bound proof is interesting, in that it adds to a growing body of recent work on
fine-grained complexity in computational geometry, and more specifically illustrates the power of the Combinatorial Clique Hypothesis.

Our near-$O((mn)^{d/2})$ upper bound also applies to the variant of the problem for \emph{undirected} Hausdorff
distance, since the undirected version of Problem~\ref{prob0}
can also be reduced to Problem~\ref{prob2}(a) with $N=O(mn)$.
However, more effort might be needed to adapt our lower bounds to
the undirected problem (although we have not tried seriously).

For noncombinatorial algorithms, our reduction implies a lower bound of $\Omega((mn)^{d\omega/6-\delta})$, under the standard hypothesis that the $k$-clique problem for graphs with $\nG$ vertices requires $\Omega(\nG^{d\omega/3-\delta'})$ time,
where $\omega\in [2,2.373)$ denotes the matrix multiplication exponent.
Proving better conditional lower bounds for noncombinatorial algorithms remains open.
This might require further new techniques, as we currently do not have
tight conditional lower bounds for the original Klee's measure problem for noncombinatorial
algorithms for $d\ge 4$~\cite{Kun22}.

As mentioned, Bringmann and Nusser~\cite{BringmannN21} proved a near-$mn$ lower bound for $d=2$ 
under the OV Hypothesis; their result is in some sense
more robust (it holds for noncombinatorial algorithms) and 
applies also to the $L_2$ case (and $L_p$ for any $1\le p\le\infty$). However,
the problem for $L_2$ probably has higher complexity than for $L_\infty$: the best upper bounds are near $n^3$ for $d=2$ and near $n^5$ for $d=3$~\cite{HuttenlocherKS93}, and
near $n^{\up{3d/2}+1}$ for $d\ge 4$~\cite{ChewDEK99}, in the $m=n$ case.
(See Bringmann and Nusser's paper for a 3SUM-based lower bound for the $L_2$ problem for $d=2$
in the ``unbalanced'' case when $m$ is constant.)

\PAPER{\small}
\bibliographystyle{plainurl}
\bibliography{pt_match}

\end{document}